\documentclass[onecolumn]{article}
\author{Ka.Shrinivaasan \\ Chennai Mathematical Institute (CMI), Chennai, India \\ (shrinivas@cmi.ac.in) \\ \\ advised by \\ Dr.B.Ravindran, Indian Institute of Technology (IIT), Chennai \\ and \\ Dr.Madhavan Mukund, CMI \\ (ravi@cse.iitm.ac.in, madhavan@cmi.ac.in)}
\title{Few algorithms for ascertaining merit of documents and their applications}
\usepackage[T1]{fontenc}
\usepackage{ucs}
\usepackage{amsmath}
\usepackage{amsthm}
\usepackage{amsfonts}
\usepackage{amssymb}
\usepackage[pdftex]{color,graphicx}
\newtheorem{defn}{Definition}
\newtheorem{theorem}{Theorem}
\begin{document}
\maketitle
\newpage
\begin{abstract}
Existing models for ranking documents(mostly in world wide web) are prestige based. In this article, three algorithms to objectively judge the merit of a document are proposed - 1) Citation graph maxflow 2) Recursive Gloss Overlap based intrinsic merit scoring and 3) Interview algorithm. A short discussion on generic judgement and its mathematical treatment is presented in introduction to motivate these algorithms.
\end{abstract}

\section {Introduction and Motivation}
Motivation for objective, independent judgement of a document is founded on the following example: 
\begin{quote}
Judge X decides about the merit of an entity Z purely by what other entities opine about Z without interacting with Z; Judge Y decides about the merit of Z by interacting only with Z. Question now is who is better judge - X or Y. 
\end{quote}
Probability of judgmental error of judge X is equal to probability of collective error of entities opining about Z while probability of judgemental error of judge Y is 0.5 as the following elementary arithmetic shows. Let us assume there are $2n$ voters and they need to decide/vote on  whether a candidate is good or bad. A candidate getting majority ( $n + 1$ good votes) will be winner.
\begin{quote}
Question: What is the probability that people have made a good decision? 
\end{quote}
\begin{quote}
Answer: Probability of each voter making a good decision is $p$ and bad decision is $1 - p$ ( 0 <= p <= 1). Let p = 0.5 for an unbiased voter.
\end{quote}
So for a candidate to be judged 'good', atleast $n + 1$ people should have made a good decision. Probability of a good choice for these $2n$ voters, skipping the calculations, is :
\begin{multline}
P(good) = ((2n)! / 4^n ) * (( 1/((n+1)!(n-1)!) + \\
			1/((n+2)!(n-2)!) \\
			+ ...... + 1/((n+n)!(n-n)!)))
\end{multline}
If there is an objective judgement without voting, probability of good decision is 0.5. It is interesting to see that above series tends to 0.5 as n grows infinitely. Thus, the judgement-through-majority-vote error probability is equal to the error probability of judge X who uses only the inputs from witnesses to judge Z while judgement-through-interaction(without election) error probability is equal to the error probability of judge Y (i.e. 0.5) who does not use witnesses. Thus, both judges X and Y are equally fallible but the cost incurred in a real world scenario for simulating X far outweighs that of Y.Thus it is worth delving into schemes for objective judgement like Y. 
\section {Three algorithms presented hereunder}
\begin{enumerate}
\item Maxflow and Path lengths of Citation graphs - objective judgement (differs from Pagerank since it is Maxflow based and not prestige based)
\item Generalized Recursive Gloss Overlap - objective judgement (simulates judge Y with a 'white-box', invasive, intrinsic merit scoring) - covers majority of this report
\item Interview algorithm - objective judgement (simulates judge Y; Uses questions and answers to judge a candidate - 'black-box' and less-invasive - and also incorporates intrinsic merit score obtained from either MaxFlow of Citation graph or Generalized Recursive Gloss Overlap)
\end{enumerate}
\section {Directed Graph of Citations}

\subsection {Average Maxflow and Path lengths of Directed Graph of Citations} 
Given a corpus, algorithm constructs directed graph of incoming links to a document x from those documents chronologically later than x.Thus corpus is partitioned into set of digraphs. Indegree of a vertex in this digraph reflects the importance of a document represented by a vertex. This digraph can be thought of as a flow network where concept flows from a document to others which cite. Each edge has a weight. Flow/weight for an (u,v) edge is defined as number of references v makes while citing u though there could be other ways to weight an edge. Assigning polarity to this weight is discussed in 4.2. Mincut of the digraph is the set of documents which are "potentially most influenced by the source document" (because maximum flow of concept from source occurs through this set to outside world/sink). Thus size of maxflow/mincut, averaged over all vertex-pairwise maxflow values, is a measure of influence of a source document in a community and thus points to its merit. (E.g., Chronology for web documents can be found by 'Last-modified' HTTP header which every dynamic document server is mandated to send to client). Alternative way to get the merit is to count the number of vertices in a predefined radius from source (i.e set of paths of some fixed length from source) which can be less accurate and sometimes misleading. Thus documents can be ranked using average Maxflow values. Advantage of this scheme is that it quantifies the extent of percolation of a concept within a community through Maxflow, without giving importance to the prestige measure of the vertices(documents) involved. So, this is one way of objectively assessing the merit of a vertex(document). Implementation applies Ford-Fulkerson algorithm to each \(s,t\) distinct pair and finds the average maxflow out of each vertex.

\subsection {Polarity of citation edge}
Parse the document/sentence containing the citation/link into tokens and find polarity. Whether a word is positive or negative can be decided by:
\begin{enumerate}
\item looking up a sentiment annotated ontology (e.g positivity/negativity of a lemma in SentiWordNet) or 
\item entropy analysis - using $\sum_{i=0}^1(-P(i)logP(i))$ where P(0) = percentage of positive words and P(1) = percentage of negative words. Closer the entropy to zero, clearer the sentence/document on its viewpoint (very good or very bad) or 
\item recursive gloss overlap algorithm to the citing document to get the polarity/sentiment of context citing the document. 
\end{enumerate}
Implementation tries all the three above. If the polarity/sentiment is negative, the weight for edge (u,v) is made negative in citation digraph, indicating a negative flow of concept to vertex v from the cited vertex u.
\section{Definition Graph Convergence(or)Generalized recursive gloss overlap}
\subsection {Motivation for computing Intrinsic Merit of a document}
Intrinsic merit is defined as the amount of intellectual effort put forth by the reader of a document and we try to quantize this effort. It is important to note that this quantized effort is independent of any observer/link-graph. Any document goes through some human understanding and we try to model it through what can be called Iceberg/Convergence/Generalized recursive gloss overlap algorithm (named so because a web document contains only a tip of the knowledge a document represents and understanding the document requires deeper recursive understanding of the facts or definitions the document is home to.).For example, going through a research paper requires the understanding of the concepts which draw a logical graph in our mind. Thus time spent on grasping the concepts and hence the intrinsic merit is proportional to the size and complexity of this graph and points to its merit (which is equal to the intellectual effort of the human reader). Since WordNet is the existing model for semantic relationship, we will try to establish that a text document can be mapped to a graph which is a subgraph of WordNet and merit can be derived applying some metrics on this graph. This is the intuition behind the algorithms that follow.

\subsection {Definition tree of a document}
Given a document its definition tree is recursively defined as

\begin{defn}
definitiontree(all keywords of document) 
= definitiontree(term1)
   definitiontree(term2)
...definitiontree(termn) where term1, term2,...termn 
occur in the definition of keywords of a document.
\end{defn}
For example, let us consider the following document which talks about Kuratowski theorem

\begin{quote}
Document1 = Every K5,K3,3-free graph is planar
\end{quote}

This document contains key terms like "K5,K3,3-free", "graph" and "planar".Now we recursively construct the definition tree for these terms. Key terms are decided after filtering out stopwords and by computing TF-IDF and only terms above a threshold tfidf are chosen for constructing the definition graph.

definitions at level 1:
\begin{enumerate}
\item K5 = Complete graph of 5 vertices (key terms: graph, vertices)
\item K3,3 = graph of two sets of 3 vertices each interconnected (key terms: graph, two sets, vertices, interconnected)
\item graph = set of vertices and edges among them (key words: vertices, edges, set)
\item planar= graph embedded on a plane (key words: graph, embedded, plane)
\end{enumerate}

Thus the definition tree goes deeper as each keyword/concept is dissected and understood. Given above is level-1 grasping of the document. Important thing to note is that intersection of the sets of keywords in the definition of K5, K3,3, graph and planar is not an empty set(glosses for two or more keywords overlap). For example, intersection of definitions of K5 and K3,3 is the set \{graph, vertices\}. Thus the overlap of the terms "graph" and "vertices" in two definitions of K5 and K3,3 is an indication of deeper cohesion/interrelatedness of the terms in the document.Thus the replicated terms(represented by vertices) in the definition tree can be merged to get convergence (gloss overlap generalized to more than two glosses). Thus the definition tree is transformed into definition graph(since a vertex can have more than one parent) by merging replicated keyword vertices into 1 vertex.
Synset definitions in WordNet gloss are used for getting keyword definitions in the implementation. But WordNet Gloss does not work for terms specialized for a domain(e.g gloss for "graph" does not have a synset for graph theory as part of its senses set). This requires ontologies for the class the document belongs to. Thus recursive gloss overlap algorithm is limited by WordNet in present implementation.
At each level, word sense disambiguation is done by following Lesk's algorithm adapted to Generalized Recursive Gloss overlap to choose the synset definition fitting the context.
It is important to note that 1) only one relation ("is in definition of") is used and 2) only keywords within the document are considered 3) gloss overlap is computed recursively at each level of understanding till required depth is reached.

\subsection {Definition graph convergence and steps of Recursive Gloss Overlap algorithm}

Convergence of a document is defined as the decrease in the number of unique vertices of the set of definition trees of its keywords from level k to level k+1
For example definition tree of the above document converges to \{edges, vertices\} after expanding the definition tree further down. Thus the above document has "edges" and "vertices" as its undercurrent. Thus the Convergence algorithm takes no labelled examples for inference. Only requirement is to have a dictionary/gloss/ontology of terms and their corresponding definitions.
If a documents definition tree does not converge within a threshold called "depth" number of levels then the document is most likely less meaningful or of low merit. Thus the Convergence algorithms strikingly adapts an iceberg which has seemingly unconnected set of "tips" at the top but as we go deeper get unified. Level where this unification happens is a differentiator of merit.
If while recursively expanding the definition tree, a vertex results in a child vertex which is same as some sibling of the parent then we compute and
remove the intersection of keywords at present and previous level - since these common vertices have already been grasped. Accordingly, number of edges, vertices and relatedness are updated for each level. Number of vertices are adjusted for removal of common tokens, but number of edges remain same since they just point to a different vertex at that level. This process continues top-down till required depth is reached.

Steps:
\flushleft
\begin{enumerate}
\item Get the document as input
\item currentlevel = 1
\item keywordsatthislevel = \{keywords from the document through tfidf filter ( e.g > 0.02)\}
\item While (currentlevel < depthrequired) \{
\begin{itemize}
\item For each keyword from keywordsatthislevel lookup the best
matching definition for the keyword and add to a set of tokens in next level - requires WordSenseDisambiguation - implementation uses Lesk's algorithm
\item Remove common tokens with previous levels since they have been grasped in previous level (this is an optimization)
\item Update the number of vertices, edges and relatedness (vertices correspond to unique tokens, edges correspond to the single relation 'y is in definition of x' and relatedness is linear overlap or quadratic overlap) and Update tokensofthislevel
\item currentlevel = currentlevel + 1
\end{itemize}
\} 
\item Output the Intrinsic merit score =\[ |vertices|\cdot| edges|\cdot|relatedness| / firstconvergencelevel\]
Where
\begin{itemize}
\item Relatedness =$NumberOfOverlaps$ (linear, also called as convergence factor) (or)
Relatedness = $NumberOfOverlappingParents * NumberOfOverlaps^2$ (quadratic)
\item firstconvergencelevel = level of first gloss overlap
\end{itemize}
\end{enumerate}

At the end of recursive gloss overlap, nodes with high number of indegrees(parents) are indicators of the class of the document since greater the indegree, greater is the number of keywords overlapping (voting for an underlying theme).From graph theoretic view, Definition Graph constructed above is a multipartite graph since vertices can be partitioned into sets with no edges within a set and edges only across sets(without removal of common tokens between levels - which is only an optimization  since by removing common tokens we redirect edges to vertices within the same set and multipartiteness is lost). Preserving multipartiteness is useful since it groups the tokens at each level of recursion into single set with edges across these sets - multipartite cliques of this multipartite graph can be analyzed to get the robustness.  Moreover, this algorithm ignores grammatical structure. Reason is that principal differentiator in analyzing relative merit of two documents is the quality of content and complexity of content and both documents are equally grammatical.Quality of content is proportional to the vertices of the definition graph and complexity of the content is proportional to the relatedness and edges of definition graph.
In spite of ignoring grammatical structure, the graph constructed above is context-sensitive since word sense disambiguation is done while choosing the synset matching a keyword. This way, the definition graph is a graph representation of the knowledge in the document sans the grammatical connectives.

\subsection {Definition of shrink}
\begin{defn}
Let us define "shrink" to be the amount of decrease in the number of unique vertices between levels $k$ and $k+1$ during convergence (gloss overlap)
\end{defn}

\subsection {Comparison of two documents for relative merit - two examples}
\begin{quote}
Document1 : Car plies on sky
\end{quote}
Constructing definition graph for level-1 we get,
\begin{enumerate}
\item Car - automobile used for surface transport
\item plies - is flexible; goes on a surface; moves
\item sky - atmosphere; not on earth;
\end{enumerate}
As can be readily seen there is overlap of 2 key terms at level 1 of the tree and thus there is less gloss overlap. Thus at level-1 document looks less meaningful.

\begin{quote}
Document2 : Cars and buses ply on road
\end{quote}
Constructing definition graph for level-1 we get,
\begin{enumerate}
\item Car - automobile used for surface transport
\item Buses - automobile used for surface transport
\item ply - flexible; go on a surface; move
\item road - asphalted surface used for transport
\end{enumerate}
All 4 keywords overlap giving {surface} as common token in their respective glosses. Overlap is better than Document1, since more keywords contribute to overlap.

\subsection {Intrinsic merit score, Convergence factor and Relatedness}
\begin{defn}
Let us define Intrinsic merit I to be the product of number of vertices(V), number of edges(E) and Convergence factor(C) of the definition graph of the document.
\begin{equation}
I = V*E*C 
\end{equation}
\end{defn}
Convergence factor (C) is the difference between number of vertices in definition tree and number of vertices in definition graph (V). Number of vertices in definition tree includes overlapping vertices without coalescing them ( since after coalescence we get the definition graph).Number of vertices in the definition tree = $x^d - 1$  where x is the average number of keywords per term definition and d is the depth of the definition tree of the document. Let us add 1 to this to get $x^d$ (smoothing).
Number of vertices in the definition graph = V
Thus the Convergence factor C and Intrinsic merit I become,
\begin{equation}
C = x^d - V
\end{equation}
\begin{equation}
I = V*E*(x^d - V)
\end{equation}
Intrinsic Merit score can also be further fine-tuned by taking into account the level of definition tree at which first convergence(gloss overlap) happens, defined as firstconvergencelevel. Greater the firstconvergencelevel, more irrelevant the document "looks" (but has a deeper cohesion). Depth to which definition tree has to be grown is decided by extent of grasp needed by the reader. Thus greater the depth of definition tree, greater is the understanding. It is obvious to see that Depth has to be greater than firstconvergencelevel so that some pattern can be mined from the document. Heuristically, we can grow the definition tree till intersection of leaves of all sub-trees of the keywords in the document is non-empty. This is the point where we can safely assume that all keywords in the document have been somehow related to one another.
So, Intrinsic merit score can be improved by incorporating firstconvergencelevel denoted by f. Thus improved score is
\begin{equation}
I = V \cdot E \cdot (x^d - V) / f
\end{equation} 
(since merit is inversely proportional to firstconvergencelevel) .Complexity of constructing definition tree is O($x^d$). Since non-unique vertices are coalesced(through gloss overlap), definition graph can be constructed in O(V) time (subexponential).
Since x is the average number of children keywords per keyword, $x= E/V$. Substituting,
\begin{equation}
I = E*V*(E^d-V^d) / (V^d* f) 
\end{equation}
As an alternative to convergence factor, gloss relatedness score similar to the one discussed by Banerjee-Ted, but considering only one relation, number of overlapping parents and length of overlap can be used to get the interrelatedness/cohesion of the document. Replacing the convergence factor with relatedness,
Intrinsic merit becomes, $I = V \cdot E \cdot Rel / f $ where $Rel$ is the sum of relatedness scores, computed over all overlapping glosses at each convergence level and f is the level at which first gloss overlap occurs 
\begin{equation}
Rel= \sum_{i=1}^n(relatedness(Level(i), keyword1,keyword2,...,keywordn))
\end{equation}
This relatedness score has been generalized to overlap of more than two glosses with single relation R (R(x,y) = y is in definition of x).
Function relatedness() for n-overlapping keywords is defined as,
\begin{multline}
relatedness(Level(i), keyword1, keyword2,...\\
 , keywordn) = OverlapLengthAtLevel(i) \\
(Linear Overlap) 
\end{multline}
(or)
\begin{multline}
relatedness(Level(i), keyword1, keyword2,...\\
 , keywordn) =n\cdot(OverlapLengthAtLevel(i) ^2)\\
(Quadratic Overlap)
\end{multline}
The relatedness score reflects the convergence since it takes into account the overlapping keywords at each level and length of the overlap. Thus first version of relatedness() function, implies the convergence factor (difference in number of vertices of definition tree and definition tree, signifying overlap)
Intrinsic merit/Relatedness score can be used to rank the set of documents and display them to the user.
Referring back to examples in 5.5, quadratic relatedness measure ((9) above)
is a better choice than linear overlap since it is a function of both overlapping parents and the overlap length. The quadratic overlap gives greater weightage to length of overlap by squaring it while keeping the number of parents involved linear.

\subsection {Intuition captured by above intrinsic merit score}
The number of edges (representing relation between parent term and its definitions) increase as relationship among vertices of definition graph increases. The number of vertices(keywords) in the definition graph increases, as the knowledge represented by the document increases. The depth of the definition tree increases, as the understanding grows.
Convergence factor increases as number of overlapping terms in definition graph increases. Similarly quadratic relatedness score increases with number of keywords involved in overlap and the length of overlap, thus pointing to stronger semantic relationship among the keywords.
Intuitively, definition graph is WordNet(or any other ontology) projected onto the document.

\subsection {Breadth/Depth first search of definition graph and why it is not a good choice for computing merit score}
Since Breadth/Depth first search of graph can model human process of thinking, BFS/DFS algorithms can be applied to get the merit score. Since BFS/DFS algorithms run in O($V+E$) time merit score is proportional to $V+E$ - all vertices of the graph are visited in O($V+E$) time. But the drawback of this approach is that strength of underlying theme of the document and cohesion of keywords is not captured by this merit score. Since Intrinsic merit score obtained by Convergence reckons with depth and overlapping keywords, BFS/DFS merit score is discarded

\subsection {Sentiment analysis applying Recursive gloss overlap}
Recursive Gloss Overlap algorithm after few levels down the definition tree would spell out the sentiment of writer.
\begin{quote}
Example1: "That movie was fantastic; Graphics was awesome"
Keywords at level-1 of Definition graph construction:
\begin{enumerate}
\item movie - motion picture; positive
\item fantastic - good, excellent; positive
\item graphics - software technique; positive
\item awesome - good, great; positive
\end{enumerate}
\end{quote}
Overlapping terms are \{good, positive\} and large number of keywords(parents) contribute to this overlap. Thus the document is of extolling nature about some target entity. Prerequisite is a dictionary which annotates each word with the sentiment and sense of the word(Implementation uses SentiWordNet which gives positivity/negativity for each lemma).
Sentiment analysis with Recursive Gloss Overlap is applied to finding the polarity of an edge in Citation graph (See 3.2). Recursive Gloss Overlap algorithm is applied to each Citation context and a definition graph is constructed. Keyword vertices with more than one indegree are then tested for positivity and negativity using SentiWordNet. If majority of these is positive then polarity for citation edge is positive, otherwise negative.

\subsection {False negatives}
Convergence algorithm never assigns lower merit score to a document which deserves a higher merit since a document with higher merit explains the concept with more depth/cohesion than document with lower merit. So false negatives do not exist

\subsection {False positives}
False positives exist since both a document and its arbitrarily jumbled version will get same merit score. This is prevented by assuming grammatically correct documents or by preprocessor which does parts of speech parsing to validate the grammatical structure of the document.

\subsection {Normalization}
Intrinsic merit can be compared only if the compared documents are of same class. Thus 2 documents explaining special relativity can be compared while a document on journalism can not be compared with a document on special relativity. Intrinsic Merit scores can be normalized by,
\begin{equation}
Normalized Intrinsic Merit Score = Score / Maximum Score
\end{equation}

\subsection {Ordering and Relative Merit}
\begin{defn}
Document1 is more meritorious than document2 if
\begin{enumerate}
\item document1 has more keywords that need to be understood than those of document2,
\item cohesion/interrelation of the keywords in document1 is more than that of document2,
\item average number of keywords per definition is greater for document1 than document2,
\item firstconvergencelevel(level at which first gloss overlap occurs) of document1 is less than that of document2 and
\item depth of definition tree of document1 is greater than that of document2.
\end{enumerate}
If we want a weaker definition of the above, ranking may be a partial order(where some pairs of documents may not be comparable) than a total order.This appeals to intuition since document1 may be better in some aspects but worse in some other relative to document2
\end{defn}
\subsection {Semantic relatedness or Meaningfulness of a document}
\begin{defn}
A document is meaningful to a human reader if any pair of keywords in the document are within a threshold WordNet distance e.g Jiang-Conrath distance
\end{defn}
\subsection {Formal proof of correctness of Convergence and Intrinsic Merit Score}
\begin{theorem}
If a document lacks merit, convergence(or gloss overlap) does not occur
(Corollary: Document's merit is measured by extent of convergence)
\end{theorem}

\begin{proof}
By "meritorious" document, we imply a document which is meaningful as per the definition of meaningfulness above(i.e. keywords in a document are separated within threshold WordNet distance metric like Jiang-Conrath distance). Let us denote R as a relation "is descendant of". If xRy then y is in (gloss)definition tree of keyword x(i.e y is descendant of x). If definition trees of keywords of the document are disjoint, then there is no y such that xRy and zRy for two keywords x and z. Let us define the relation S to be "two keywords are related". xSz iff xRy and zRy for some y. Thus we formalise cohesiveness/meaningfulness of a document in terms of definition graph. If a document is not meaningful then there exist no x and z such that xSz, which implies that for no y, xRy and zRy. Thus there exist no vertex y which is in definition tree of two key words. Thus convergence is a necessary condition for merit. The relation S implies that there exists a path between two keywords x and y in the document, through some intermediate nodes which are in the definition/gloss tree of x and y. There exists a threshold WordNet distance greater than length of this path since the length is finite and whether a document is meaningful depends on this threshold. Thus convergence(generalized gloss overlap) implies meaningfulness of a document as per the definition above. Moreover Intrinsic merit increases with number of edges and relatedness() - linear or quadratic. So with greater relatedness() and more number of vertices and edges, overlaps and number of nodes involved in overlap increase. This in turn implies that more number of paths are available amongst the keywords of the document since every overlap acts as a meeting point of two keyword definition trees. Probability that lengths of these paths are less than threshold WordNet distance is inversely proportional to firstconvergencelevel(level of first gloss overlap). Thus intrinsic merit score discussed earlier captures this notion.
\end{proof}

\subsection {Parallelizability}
Recursive gloss overlap is parallelizable by partitioning the tokens at each level and assigning each subset to different processors (Map) to get the tokens for next level. Individual results from processors are merged (Reduce) to get the final set of tokens for a level. This is repeated for all levels. MapReduce can be applied for parallelism.

\section {Interview Algorithm (applying (1) and/or (2) for computing intrinsic merit)}

\subsection {Motivation for Interview algorithm}

Here we map the real world scenario of an interview being conducted on a candidate where a panel asks questions and judges the candidate based on the quality of answers by candidate - candidate is a document and it is "interviewed" by a reference set of authorities.
Each document x is interviewed/evaluated by set of reference documents which will decide on the merit of the document x. Reference set initially consists of n user chosen authorities on the subject. Interview is set of queries made by reference set on the document and evaluating the answer to the queries. If
x passes the interview it is inducted into reference set. Next document will be interviewed by n+1 documents including last selected document and so on.
Hierarchy of interviews can be built. For example Document x interviews documents y and z. Document y interviews w and document z interviews p. Thus we get a tree of interviews (it could be a directed acyclic graph too, if a candidate is interviewed by more than one reference, one of which itself was a candidate earlier). The interview scores can be weighted and summed bottom-up to get the merit of the root (Analogy: hierarchy in an organization).

\subsection {Steps of the Interview algorithm}

\begin{enumerate}
\item Relevance of the document to the reference set is measured by a classifier (NaiveBayesian or SVM or search engine results for a query)
\item Intrinsic merit score of the document is computed either by Recursive Gloss Overlap algorithm (measures the meaningfulness/sanity of the candidate) (or) by citation digraph
\item Reference set interviews the candidate and gets the score
\item Value addition of the candidate document is measured (what extra value candidate brings over and above reference set)
\item Candidate is inducted into reference set based on the above criteria if candidate is above a threshold.
\end{enumerate}

\subsection {Mathematical formulation of an interview}

Interview is abstracted in terms of a set of tuples, where each tuple is of the form 
\begin{equation}
t(i) = (question, answer, expectedanswer, score)
\end{equation}
for question i.
\begin{equation}
Interview (I) = \{t(1), t(2), t(3), ..., t(n)\}
\end{equation}
\begin{equation}
t(i).score = PercentageOfMatch(t(i).answer, t(i).expectedanswer)
\end{equation}
\begin{equation}
if (\sum_{i=1}^n (t(i).score)) > referencethreshold
\end{equation}
then induct the document into reference set.In the Information Retrieval context, a question is a query and the answer is the context within the document that matches the query. The answer returned by the document is then compared with expectedanswer. Comparison is done by Jaccard coefficient of shingles (n-grams)
\begin{multline}
t(i).score = \\
|shingle(answer) \bigcap shingle(expectedanswer)| / \\
 | shingle(answer) \bigcup shingle(expectedanswer)| 
\end{multline}
\begin{enumerate}
\item Supervised: In supervised setting, each reference document is pre-equipped with user-decided set of queries and answers it expects. Thing to note is that a document is made a live object - it both has content and questions it intends to ask(set of search queries). Alternative way to compute t(i).score is to find out the definition graph of answer and expectedanswer and compute the difference between the two graphs(e.g edit distance). Downside of this is the assumption of pre-existence of correct answers which makes this a supervised learning.
\item Unsupervised: In the absence of reference questions and answers, questions that a document "intends" to ask can be thought of as set of queries for which the document has better answers(results). These set of queries/results (questions and answers) can be automatically obtained from a document through an unsupervised way by computing set of more likely to be important n-grams(by computing key phrases with tfidf above threshold) and the context of the n-grams in the reference documents.These n-grams/contexts can later be used as "references questions" (n-grams) and "reference answers"(contexts of the corresponding n-grams) to the candidate document. Thus we compensate for lack of reference Questions and Answers. Alternatively, an interview can be simply considered as the percentage similarity of definition-graph(reference) and definition-graph(candidate) obtained by edit distance.
\end{enumerate}
\subsection {Searching for answer to a query within the document (as implemented)}

If a document describing tourist places is given and the query is "What are the good places to visit in this city?", then query is parsed into key words like "good", "places", "visit" and "city" and matching contexts within the document are returned where context is the phrase of length $2n+1$ (from $x-n$ to $x+n$ locations with location of keyword being x).

\subsection {Value addition measure}

Recursive gloss overlap algorithm gives the definition graph of the candidate document. To measure the value addition we can run the recursive gloss overlap algorithm on the reference set to get the definition graph of reference set and find out the difference between the two definition graphs - reference and candidate. Since value addition is defined as the value added which is not already present, extra vertices and edges present in candidate but absent in the reference set are a measure of value addition. Value addition can be measured by either edit distance(cost of transforming one graph to the other after adding/deleting vertices/edges), maximum common subgraph or difference of adjacency matrices. Implementation uses graph edit distance measure.

\subsection {Update summarization through Interview algorithm (applying algorithm given in 6.2)}

Given a news summary and a candidate news to be added to summary
\begin{enumerate}
\item Label the summary as reference set.
\item Run a classifier on summary and candidate to get the class to which both belong to (or get from search engine results on a news topic)
\item If class(summary) == class(candidate) proceed further
\item Calculate intrinsic merit score of the candidate news document through recursive gloss overlap algorithm described in (2) (or) from citation digraph described in (1)
\item Candidate news is interviewed by reference set(summary in this case)
\item Compute value addition of candidate to summary
\item Add the value added information from candidate into existing summary to get new summary (by getting cream of sentences with top sentence scores)
\end{enumerate}

\subsection {Application to Topic Detection, Link detection and Tracking}
\begin{enumerate}
\item Interview algorithm and graph edit distance measure can be applied to news topic link detection( Answers the question - Does a pair of news stories discuss same topic?). Since same news item falls under multiple topics and is changing over time, topic of a news story is in a state of flux. Given a pair of news stories (n1, n2) execute interview of n2 with n1 as reference. This interview score decreases and edit distance grows as n2 becomes more irrelevant to n1. By defining a threshold for interview and edit distance scores to belong to same topic, link detection can be achieved. It is important to note that interview score and value addition score are inversely related.
\item At any point in time, compute edit distance for all possible pairs Nx, Ny in a topic (after getting their respective definition graphs) and choose Ny which has largest edit distance to others and hence an outlier and least likely to be in the topic. Thus topic detection is achieved(Answers the question - Does this story exist in correct topic?).
\item Topic tracking can be done by constructing definition graph and finding vertices with high number of indegrees. These keywords are voted high and point to the maximum likely topic of the news story (works as an unsupervised text classifier).This process has to be periodically done since topic of a story might change and thus the definition graph will change.
\end{enumerate}
\subsection {Implementation in Python}
\begin{enumerate}
\item Get documents of same class/topic (from Reuters corpus or search engine query results) and their future references as input
\item Compute the intrinsic merit score for both documents:
\begin{itemize}
\item applying citation digraph construction (or) applying definition graph convergence(generalized recursive gloss overlap)
\item Parse into keywords and get keywords above a threshold tf-idf
\item Perform WSD using Lesk's algorithm
\item Get glosses of matching sense through wordnet api
\item get overlaps at level i, update intrinsic merit score (either using linear or quadratic overlap)
\item repeat  for sufficient number of levels defined by "depth"
\end{itemize}
\item execute interview if reference questions and answers are available (supervised) or through getting important n-grams/context described above (unsupervised) - at present restricted to 1-gram for keywords and bigrams for jaccard coefficient calculation
\item compute value addition through definition graph edit distance between reference and candidate, and get the score.
\item get percentage weighted sum of intrinsic merit, interview and value addition scores and get final score.
\item APPLY (2), (3) and (4) ABOVE TO NEWS UPDATE SUMMARIZATION: If final score is above threshold, update the news summary with candidate news and publish (sentence scoring is done by sum of tfidf scores of words in a sentence) 
\item APPLY (2) TO GET INTRINSIC MERIT RANKING: Compare the scores of the two documents from (2) and rank. (Citation graph based maxflow ranking internally uses sentiment analysis using one of 1) Recursive gloss overlap 2) SentiWordNet 3) Entropy analysis of citation context)
\item APPLY (3) and (4) ABOVE TO TOPIC DETECTION AND TRACKING
\end{enumerate}

\newpage
\section{Results}
\subsection {Results - Intrinsic Merit score with quadratic overlap for top 10 Google ranked documents for query 'data mining' sorted ascending}
Pagerank - (Document,relatedness,vertices,edges,firstconvergencelevel) - IMScore
\begin{itemize}
\item 6 - ('ThesisDemo-datamining-test6.txt', 477660, 372, 576, 1)   - 102349163520.0 
\item 10 - ('ThesisDemo-datamining-test10.txt', 139114790, 1172, 2339, 1)  - 3.81356486745e+14 
\item 8 - ('ThesisDemo-datamining-test8.txt', 310161784, 1456, 3034, 1) - 1.37014092147e+15 
\item 7 - ('ThesisDemo-datamining-test7.txt', 310161784, 1456, 3034, 1) - 1.37014092147e+15 
\item 5 - ('ThesisDemo-datamining-test5.txt', 51304180926L, 2938, 10643, 1) - 1.60423730814e+18 
\item 4 - ('ThesisDemo-datamining-test4.txt', 99651694978L, 3324, 12921, 1) - 4.27998090689e+18 
\item 9 - ('ThesisDemo-datamining-test9.txt', 133686525217L, 3186, 13468, 1) - 5.73636152749e+18 
\item 3 - ('ThesisDemo-datamining-test3.txt', 354003740698L, 3901, 18039, 1) - 2.49112924394e+19 
\item 2 - ('ThesisDemo-datamining-test2.txt', 594730534291L, 3935, 20502, 1) - 4.79801059042e+19 
\item 1 - ('ThesisDemo-datamining-test1.txt', 2753901168066L, 5832, 33386, 1) - 5.36204253324e+20
\end{itemize}

\subsection {Results - Intrinsic Merit score with quadratic overlap for top 10 Google ranked documents for query 'philosophy' sorted ascending}
Pagerank - (Document,relatedness,vertices,edges,firstconvergencelevel) - IMScore
\begin{itemize}
\item 3 - ('ThesisDemo-philosophy-test3.txt', 63840296, 1110, 2165, 1) -  1.53417807332e+14
\item 7 - ('ThesisDemo-philosophy-test7.txt', 456552729, 1234, 3041, 1) - 1.71325703153e+15
\item 5 - ('ThesisDemo-philosophy-test5.txt', 915190280, 1428, 3651, 1) - 4.77146166914e+15
\item 6 - ('ThesisDemo-philosophy-test6.txt', 1128268242, 1891, 4577, 1) - 9.76528235921e+15
\item 2 - ('ThesisDemo-philosophy-test2.txt', 2739304610L, 2033, 5316, 1) - 2.96048373426e+16
\item 10 - ('ThesisDemo-philosophy-test10.txt', 6630859968L, 2289, 6471, 1) - 9.82170869184e+16
\item 9 - ('ThesisDemo-philosophy-test9.txt', 7675201402L, 2105, 6477, 1) - 1.04644348307e+17
\item 8 - ('ThesisDemo-philosophy-test8.txt', 9692242200L, 2165, 6733, 1) - 1.41283281476e+17
\item 4 - ('ThesisDemo-philosophy-test4.txt', 14535833906L, 2553, 7920, 1) - 2.93911072979e+17
\item 1 - ('ThesisDemo-philosophy-test1.txt', 9611266377319L, 7552, 49449, 1) - 3.58922024377e+21
\end{itemize}

\subsection {Results - Intrinsic Merit score with quadratic overlap for human (2 judges) judged documents on topic 'democracy' - sorted ascending}
Human ranking - (Document,relatedness,vertices,edges,firstconvergencelevel) - IMScore
\begin{itemize}
\item {5,5} - ('ThesisDemo-democracy-test2.txt', 15535, 270, 406, 1) - 1702946700.0
\item {4,6} - ('ThesisDemo-democracy-test6.txt', 60534, 253, 373, 1) - 5712533046.0
\item {6,1} - ('ThesisDemo-democracy-test1.txt', 136281, 249, 384, 1) - 13030644096.0
\item {2,2} - ('ThesisDemo-democracy-test4.txt', 245448, 358, 568, 1) - 49910378112.0
\item {1,3} - ('ThesisDemo-democracy-test3.txt', 1623723, 364, 671, 1) - 396584600412.0
\item {3,6} - ('ThesisDemo-democracy-test5.txt', 1167039, 485, 847, 1) - 479413786005.0
\end{itemize}

\subsection{Results - Intrinsic Merit score with quadratic overlap for human (1 judge) judged documents on topic 'soap' - sorted ascending}
Human ranking - (Document,relatedness,vertices,edges,firstconvergencelevel) - IMScore
\begin{itemize}
\item  4 - ('ThesisDemo-soap-test4.txt', 52, 212, 346, 2) - 1907152.0
\item  3 - ('ThesisDemo-soap-test2.txt', 735, 113, 146, 1) - 12126030.0
\item  2 - ('ThesisDemo-soap-test3.txt', 1368, 109, 152, 1) - 22665024.0
\item  1 - ('ThesisDemo-soap-test1.txt', 2912, 188, 251, 1) - 137411456.0
\item  5 - ('ThesisDemo-soap-test5.txt', 25641, 230, 353, 1) - 2081792790.0
\end{itemize}

\subsection {Results - Intrinsic Merit score with quadratic overlap for top 7 Google news stories for query 'haiti earthquake' - sorted ascending}
Pagerank - (Document,relatedness,vertices,edges,firstconvergencelevel) - IMScore
\begin{itemize}
\item 4 - ('ThesisDemo-haiti-test4.txt', 11683630, 710, 1343, 1) - 1.11406917139e+13
\item 2 - ('ThesisDemo-haiti-test2.txt', 65287245, 1002, 2008, 1) - 1.31358981536e+14
\item 7 - ('ThesisDemo-haiti-test7.txt', 219493417, 1258, 2785, 1) - 7.69001771262e+14
\item 6 - ('ThesisDemo-haiti-test6.txt', 491851745, 1321, 3223, 1) - 2.09409962803e+15
\item 3 - ('ThesisDemo-haiti-test3.txt', 4268535180L, 1966, 5693, 1) - 4.7775315353e+16
\item 5 - ('ThesisDemo-haiti-test5.txt', 7043167094L, 2120, 6412, 1) - 9.57408693023e+16
\item 1 - ('ThesisDemo-haiti-test1.txt', 44329850203L, 3052, 10603, 1) - 1.434529734e+18
\end{itemize}

\subsection {Results - Intrinsic Merit score with quadratic overlap for top 10 Google ranked documents for query 'literary' sorted ascending}
Pagerank - (Document,relatedness,vertices,edges,firstconvergencelevel) - IMScore
\begin{itemize}
\item 5 - ('ThesisDemo-literary-test5.txt', 38252283, 1032, 1944, 1) - 7.67420361729e+13
\item 7 - ('ThesisDemo-literary-test7.txt', 815611695, 2020, 4386, 1) - 7.22609124643e+15
\item 3 - ('ThesisDemo-literary-test3.txt', 5989035631L, 2039, 5938, 1) - 7.25127400033e+16
\item 10 - ('ThesisDemo-literary-test10.txt', 6155411625L, 2467, 6713, 1) - 1.01939593415e+17
\item 8 - ('ThesisDemo-literary-test8.txt', 296376674293L, 4598, 18333, 1) - 2.4983111474e+19
\item 4 - ('ThesisDemo-literary-test4.txt', 529359994275L, 5074, 21758, 1) - 5.84413920691e+19
\item 2 - ('ThesisDemo-literary-test2.txt', 643163471944L, 4920, 22433, 1) - 7.09861839373e+19
\item 1 - ('ThesisDemo-literary-test1.txt', 1149789557857L, 5126, 25916, 1) - 1.52744272126e+20
\item 9 - ('ThesisDemo-literary-test9.txt', 2149531315027L, 6056, 31756, 1) - 4.13385687561e+20
\item 6 - ('ThesisDemo-literary-test6.txt', 3388627800057L, 6826, 36617, 1) - 8.4697952824e+20
\end{itemize}

\subsection {Results Excerpt- Intrinsic Merit Ranking of Reuters corpus in 'earn' category}
(Document,relatedness,vertices,edges,firstconvergencelevel) - IMScore
\begin{itemize}
\item ('test/15046', 34, 59, 93, 1) - 186558.0
\item ('test/14911', 55, 164, 219, 1) - 1975380.0
\item ('test/15213', 65, 169, 234, 1) - 2570490.0
\item ('test/15063', 105, 226, 331, 2) - 3927315.0
\item ('test/14899', 79, 199, 278, 1) - 4370438.0
\item ('test/15070', 107, 199, 306, 1) - 6515658.0
\item ('test/15185', 117, 210, 327, 1) - 8034390.0
\item ('test/15074', 116, 219, 335, 1) - 8510340.0
\item ('test/15103', 107, 259, 366, 1) - 10142958.0
\item ('test/14965', 125, 232, 357, 1) - 10353000.0
\end{itemize}
\subsection {Results - Spearman ranking coefficient and Pearson coefficient for the above rankings}
\begin{itemize}
\item (1) . Spearman ranking coeffiecient for Google ranking for 'data mining'  :  0.733333333333
\item (2) . Pearson ranking coeffiecient for Google ranking for 'data mining'  :  0.00888888888889
\item (1) . Spearman ranking coeffiecient for Google ranking for 'literary' :    0.0909090909091
\item (2) . Pearson ranking coeffiecient for Google ranking for 'literary' :    0.00110192837466
\item (1) . Spearman ranking coeffiecient for Google ranking for 'philosophy' :  0.0424242424242
\item (2) . Pearson ranking coeffiecient for Google ranking for 'philosophy' :  0.000514233241506
\item (1) . Spearman ranking coeffiecient for Google ranking for 'haiti earthquake'  :  0.25
\item (2) . Pearson ranking coeffiecient for Google ranking for 'haiti earthquake' :  0.00892857142857
\item (1) . Spearman ranking coeffiecient for Human ranking(2 judges) for 'democracy' :  0.385714285714
\item (2) . Pearson ranking coeffiecient for Human ranking(2 judges) for 'democracy' :  0.0174334140436
\item (1) . Spearman ranking coeffiecient for Human ranking(1 judge) for 'soap' :  0.9
\item (2) . Pearson ranking coeffiecient for Human ranking(1 judge) for 'soap' :  0.09
\end{itemize}
As per Spearman coefficient, correlations between Google ranking and Recursive Gloss Overlap are 73\%, 4\%, 9\% and 25\% while with human ranking they are 38\% and 90\%

\subsection {Results - Update Summarization applying Interview algorithm 
with Recursive Gloss Overlap - Example - Summary size is 12.5\%}
Candidate document:
\begin{quote}
1 Dead in Bangkok Protests, More Than 70 Wounded

VOA News22 April 2010
A Thai woman lies injured on the ground in Bangkok after several small explosions occurred near site of anti-government protests in Bangkok, 22 Apr 2010
Photo: AP

A Thai woman lies injured on the ground in Bangkok after several small explosions occurred near site of anti-government protests in Bangkok, 22 Apr 2010

Hospitals in Thailand's capital, Bangkok, say at least one person has been killed and many others wounded in a series of explosions near an encampment of anti-government protesters.

More than 70 people were reported wounded Thursday at the camp in the city's business district, which is packed with armed troops and diverse groups of protesters.  Reports say at least five hand grenades exploded, prompting the closure of a nearby train station.

The protesters have besieged central Bangkok for weeks, trying the patience of citizens and business leaders.  A coalition of groups gathered to drive the so-called Red Shirt protesters from the main retail and tourist district.  Police separated the two sides.

The Red Shirt protesters have rallied for five weeks, demanding new elections and the resignation of Prime Minister Abhisit Vejjajiva.

A coalition opposed to the Red Shirts - called the Multi-Colored Shirts - has announced a mass rally for Friday.

Most of the protesters support former Prime Minister Thaksin Shinawatra, who was ousted in 2006.  Mr. Thaksin lives in exile and faces a prison sentences on corruption charges in Thailand.  He has a significant following among the country's rural and low-income population.

Mr. Abhisit came to power in December 2008, after months of massive anti-Thaksin protests by demonstrators known as the Yellow Shirts.   

The military said soldiers will use tear gas, rubber bullets and live ammunition, if necessary, to remove the protesters.  But Army Chief General Anupong Paochinda has been reluctant to use arms fearing renewed bloodshed.  

An April 10 clash between the Red Shirts and Thai security forces left at least 25 people dead, and 850 others injured.
\end{quote}
Reference document:
\begin{quote}
Bangkok blasts kill one, injure 75 - Thai media
11:16am EDT

BANGKOK, April 23 (Reuters) - A series of grenade blasts that rocked Bangkok's business district on Friday killed at least one person and wounded 75, hospitals and the Thai media said.

Five M-79 grenades hit an area packed with heavily armed troops and studded with banks, office towers and hotels. Four of the wounded had serious injuries, including two foreigners, according to witnesses, hospital officials and an army spokesman. (Additional reporting by Nopporn Wong-Anan; Writing by Jason Szep; Editing by Bill Tarrant)

\end{quote}
Summary:
\begin{quote}
\emph {1 Dead in Bangkok Protests, More Than 70 Wounded}

\emph {VOA News22 April 2010}
\emph {A Thai woman lies injured on the ground in Bangkok after several small explosions occurred near site of anti-government protests in Bangkok, 22 Apr 2010}
\emph {Photo: AP}

\emph {A Thai woman lies injured on the ground in Bangkok after several small explosions occurred near site of anti-government protests in Bangkok, 22 Apr 2010}

\emph {Hospitals in Thailand's capital, Bangkok, say at least one person has been killed and many others wounded in a series of explosions near an encampment of anti-government protesters . Bangkok blasts kill one, injure 75 - Thai media}
\emph {11:16am EDT }

BANGKOK, April 23 (Reuters) - A series of grenade blasts that rocked Bangkok's business district on Friday killed at least one person and wounded 75, hospitals and the Thai media said .
\end{quote}

\subsection {Results Excerpt - Topic Detection and Tracking applying Interview algorithm with Recursive Gloss Overlap}
Example Reference News Story (ThesisDemo-ipad-test1) - Topic 'ipad':
\begin{quote}
Apple unveils iAd platform; iPad sales look strong
Photo
6:29am EDT

By Gabriel Madway

CUPERTINO, California (Reuters) - Apple CEO Steve Jobs showed off a new smartphone operating system on Thursday that features an advertising platform to compete with Google's, and revealed stronger-than-expected sales of 450,000 units for the iPad.

The iPhone 4.0 software will be available on Apple's hugely popular smartphone this summer, complete with a number of upgrades, including a long-awaited multi-tasking capability that allows the use of several applications at once.

A version of the iPhone's operating system is also used on the iPad, and the latest generation of software will come to Apple's new tablet computer this fall.

The new advertising platform for the iPhone and iPad -- dubbed iAd -- marks Apple's first foray into a small but growing market, and is sure to please the thousands of application developers who make their living off those devices, providing them with a new revenue stream.

The iPad's early sales impressed analysts, many of whom expect 1 million units to be sold in the quarter ending June, and roughly 5 million in 2010, though estimates vary widely.

"We're making them as fast as we can. Our ramp is going well, but evidently we can't quite make enough of them yet so we're going to have to try harder," Jobs said, noting iPad sellouts at Best Buy stores.

The electronics giant has staked its reputation on the 9.7-inch touchscreen tablet, essentially a cross between a smartphone and a laptop. It is helping foster a market for tablet computers that is expected to grow to as many as 50 million units by 2014, according to analysts.

"I think it's pretty impressive, five days almost half a million units, and it shows there's still pretty good momentum behind the first day," said Gartner analyst Van Baker.

Despite critics who question whether a true need exists for such a gadget, analysts expect Hewlett-Packard, Dell and others to trot out their own competing devices this year.

Since the iPad went on sale on April 3, users have downloaded 600,000 digital books and 3.5 million applications for the device, Jobs said. There are already 3,500 apps available for the iPad.

"It was above my expectations, frankly," said Joe Clark, managing partner of Financial Enhancement Group, referring to iPad sales. "The day the original Apps Store launched it was a game change for the iPhone and it will do the same eventually for the iPad."

At a media event at the company's Cupertino, California, headquarters, Jobs said Apple had so far sold more than 50 million iPhones, the smartphone that competes with Research in Motion's Blackberry and Motorola's Droid.

That implies that the company sold 7 million or more devices in the March quarter, which would be above many analysts' forecasts.

MOBILE AD WAR

Apple is expected to launch the fourth-generation model of its iPhone, which was introduced in 2007, later this year.

Pancreatic cancer survivor Jobs, looking thin but energetic, introduced the iAd mobile platform, which he said had the opportunity to make 1 billion ad impressions a day on tens of millions of Apple mobile device users.

IAds will allow applications developers to use advertisements in their apps, pocketing 60 percent of the revenue. Apple will sell and host the ads.

Jobs harshly criticized the current manner and look of mobile advertising, particularly search ads. He promised that iAds will foster more engaging advertising that will not pull users away from the content within apps.

NEW ARENA

Tim Bajarin, president of consulting company Creative Strategies, said it was a dramatic shift in thinking about the delivery of mobile ads, and an obvious move by Apple to set itself apart from Google Inc, which made its name on search ads.

"It's very clear that Jobs believes that ads in the context of apps makes more sense than generic mobile search," he said.

Apple's entry into the mobile ad arena had been widely expected. This year, it paid \$270 million for Quattro Wireless, an advertising network that spans both mobile websites and smartphone applications.

Google, which already sells advertising on smartphones, agreed to buy mobile ad firm AdMob late in 2009. U.S. regulators are examining the deal's antitrust implications.

Jobs said Apple was also in the hunt to buy AdMob before Google "snatched them from us because they didn't want us to have them." The comments were just the latest hint at the rift that has emerged between Apple and Google, which were once allies but now compete in a number of arenas.

Research group Gartner expects the mobile advertising market to expand by 78 percent to \$1.6 billion in 2010.

Jobs also said the new operating system will include support for multi-tasking -- addressing a perennial consumer complaint -- allowing users to switch between several programs running simultaneously.

Shares of Apple turned positive briefly after Jobs' announcement, before quickly dipping back into negative territory. They closed 0.3 percent lower at \$239.94 on the Nasdaq.
(Writing by Edwin Chan; Editing by Steve Orlofsky, Leslie Gevirtz and Matthew Lewis)
\end{quote}
Example Candidate News Story (ThesisDemo-dantewada-test1) - Topic 'Dantewada':
\begin{quote}
Chidambaram offers to quit, PM says no
CNN-IBN

New Delhi: Union Home Minister P Chidambaram has reportedly offered to resign taking responsibility for the Dantewada massacre in which the Maoists butchered 76 security personnel. However, Prime Minister Manmohan Singh has rejected Chidambaram's offers to resign.

Chidambaram reportedly met Prime Minister Manmohan Singh taking "full responsibility" for the Dantewada attack on a CRPF patrol party and offered to step down.

Sources say that a section within the Congress party has been unhappy with the way Chidambaram has handled the Maoist issue so far including his tough talk on the rebels and his friction with West Bengal Chief Minister Buddhadeb Bhattacharjee.

When the Home Minister came back from his Dantewada tour he reportedly met the Prime Minister and told him that if he (Prime Minister) was dissatisfied with his performance, he was ready to step down.

Earlier, on Friday while attending the Valour Day function of the CRPF, Chidambaram said, "I salute the CRPF. I promise that government will always stand by you. Where does the buck stop after Dantewada? The buck stops at my desk. I accept full responsibility of what happened in Dantewada. I told this to the Prime Minister as well."

In the past whenever there was talk of all-out action to control Maoists, it was usually tempered by the Congress party that the issue it was a socio-economic one and had to be dealt with in such a manner.

Many leaders had been maintaining that Maoists were our own people and so they should not be dealt with in the same firm way as one would deal with terrorists.

But now that stand seems to have been diluted a bit within the Congress following the massacre of the CRPF team on Tuesday.

The party, too, has begun to realise that to go soft in the weeding out Maoists will not really go down very well.

So Chidambaram's offer to step down accepting complete responsibility is seen as a political masterstroke by him and also a plus point in his favour.

A large group of Maoists had ambushed the CPRF team belonging to the 62 Battalion between 6 AM and 7 AM on Tuesday between Chintalnar and Tademetla villages in Sukma block of Dantewada district of Chhattigarh when the security personnel were on the way to Tademetla in a vehicle.

The Maoists, believed to be between 200-1000, first triggered a land mine destroying the vehicle carrying the security personnel. In the ensuing gunbattle 75 CRPF men and one local police constable were killed.

\end{quote}
Topic Link Detection(to check if above two news stories discuss same topic):
\begin{itemize}
\item Interview score:3.09090909091
\item Interview score(in percentage correctness): 10.303030303
\item Edit Distance (as percentage value addition from reference):79.4014084507
\item Topic Link Detection - ThesisDemo-ipad-test1.txt and ThesisDemo-dantewada-test1.txt do not discuss same topic 
\end{itemize}
Topic Detection(to check if a news story is under correct topic) - Topics are 'ipad'(1 story), 'sukhna'(2 stories), 'lufthansa'(1 story),'dantewada'(1 story):
\begin{itemize}
\item Topic Detection - News story ThesisDemo-ipad-test1.txt has largest pairwise editdistance from ThesisDemo-sukhna-test1.txt and least likely to be in this topic
\item Topic Detection - News story ThesisDemo-ipad-test1.txt has largest pairwise editdistance from ThesisDemo-lufthansa-test1.txt and least likely to be in this topic
\item Topic Detection - News story ThesisDemo-ipad-test1.txt has largest pairwise editdistance from ThesisDemo-sukhna-test2.txt and least likely to be in this topic
\item Topic Detection - News story ThesisDemo-ipad-test1.txt has largest pairwise editdistance from ThesisDemo-dantewada-test1.txt and least likely to be in this topic
\item Topic Detection - News story ThesisDemo-dantewada-test1.txt has largest pairwise editdistance from ThesisDemo-ipad-test1.txt and least likely to be in this topic
\end{itemize}

\subsection {Results Excerpt - Applying Sentiment Analysis with Recursive Gloss Overlap for finding polarity of an edge in Citation Graph Maxflow (1)}
\begin{itemize}
\item Nodes with more than 1 parent (and hence the most likely classes of document) are: set(['good', 'used']) 
\item getPositivity: good 
\item getNegativity: good 
\item getPositivity: used 
\item getNegativity: used 
\item negative words: [] 
\item positive words: ['good', 'used'] 
\end{itemize}

\subsection {Results Excerpt - Citation graph maxflow with a simple link graph example - polarity determined by Recursive Gloss Overlap}
\begin{itemize}
\item Following is the adjacency matrix for hyperlink graph amongst 7 html documents (file1.html, file2.html, file3.html, file4.html, file5.html, file6.html, file7.html)
\begin{equation*}
Adjacency \quad Matrix \quad of \quad Link \quad Graph = \quad
\begin{bmatrix}
0 & 1  & 1 & 1 & 1 & 1 & 0 \\
1 & 0  & 1 & 1 & 0 & 1 & 0 \\
1 & 1  & 0 & 1 & 1 & 0 & 1 \\
-1 & 1  & 1 & 0 & -1 & 1 & 0 \\
1 & 1  & 0 & 1 & 0 & 1 & 1 \\
0 & 0  & 0 & 0 & 0 & 0 & 0 \\
0 & 0  & 0 & 0 & 0 & 0 & 0 
\end{bmatrix}
\end{equation*}
\item Average concept maxflow out of each page: 
\begin{itemize} 
\item 'file2.html': 3.7142857142857144
\item 'file4.html': 3.2857142857142856
\item 'file5.html': 2.0
\item 'file3.html': 3.4285714285714284
\item 'file7.html': 0.0
\item 'file1.html': 3.4285714285714284
\item 'file6.html': 0.0
\end{itemize}
\item Number of nodes within radius 2 of the source  file2.html =  7
\item Number of nodes within radius 2 of the source  file3.html =  7
\item Number of nodes within radius 2 of the source  file4.html =  7
\item Number of nodes within radius 2 of the source  file5.html =  6
\item Number of nodes within radius 2 of the source  file1.html =  7
\item Number of nodes within radius 2 of the source  file6.html =  0
\item Number of nodes within radius 2 of the source  file7.html =  0
\end{itemize}
Average maxflow values above show that file2, file1, file3, file4, file5, file6, file7 are ranked in descending order of their Maxflow merit.Thus file2 is most influential in this community. The radius measure with length 2 gives the ranking: file2, file3, file4, file1 all tied in first and file5 ranked next.
\subsection {Conclusion}
Motivation for this excercise, as evident from results above, is to explore the possibility of finding an algorithm/framework to assess the merit of a document with and without link graph structure in place with greater emphasis on the latter. Citation graph maxflow measures the penetration of a concept (represented in a document), in a link graph while the Recursive gloss overlap objectively judges the document without getting inputs from any incoming links. Interview algorithm uses either of these two algorithms and abstracts some real world applications.As seen above Google rankings differ (with some exceptions) from Recursive gloss overlap intrinsic merit rankings which is on the expected lines - content-and-complexity based merit scoring is not necessarily same as popularity based ranking. Moreover the intrinsic ranking scheme given above need not be the only possible way of computing merit. Once we have definition graph for a document (whether multipartite or not), multitude of more ranking schemes can be invented - for example based on 1) k-connectedness of the definition graph 2) completeness or robustness of the definition graph 3) (multipartite) cliques of (multipartite) definition graph (if multipartite) etc., Since definition graph construction is computationally intensive, there is a scope of improvement in improving the recursive gloss overlap algorithm by applying some parallel processing framework like MapReduce. Applying Evocation WordNet, implementing a MapReduce(e.g Hadoop) cluster and considering more than one relation are future directions to think about. Theoretical foundation for the recursive gloss overlap comes from WordNet itself which visualises the relatedness of words - Definition graph is just an induced subgraph of WordNet for a document. Since merit quantification of a document can not be done without analyzing relatedness of keywords in a document, definition graph, which is a subgraph of WordNet for a document, is a plausible representation.  Accuracy of Recursive gloss overlap depends on the accuracy of WordNet, depth to which definition trees are grown and Word Sense Disambiguation.

\end{document}